\documentclass[letter,11pt,final]{article}
\usepackage{typearea}
\typearea{14}

\usepackage[colorlinks=true,citecolor=blue,urlcolor=blue,bookmarks=false]{hyperref}
\usepackage[numbers,sort]{natbib}
\usepackage{amsfonts} % e.g. mathbb
\usepackage{amsmath,amsthm,amssymb}
\usepackage{thmtools}
\usepackage{thm-restate}
\usepackage{cleveref}
\usepackage{xcolor}
\usepackage{xspace}
\usepackage{bm}
\usepackage{todonotes}
\usepackage{enumitem}
\usepackage{float} % forces figures exactly where you put them

\usepackage{algorithm} % for writing algorithms
\usepackage[noend]{algpseudocode}
\usepackage[labelfont=bf]{caption}

\usepackage{tikz}
\usetikzlibrary{positioning}

\newcommand\numberthis{\addtocounter{equation}{1}\tag{\theequation}}

\theoremstyle{plain}
\newtheorem{theorem}{Theorem}[section]

\newtheorem{proposition}[theorem]{Proposition}

\newtheorem{fact}[theorem]{Fact}
\newtheorem{lemma}[theorem]{Lemma}
\newtheorem{claim}[theorem]{Claim}
\newtheorem{definition}[theorem]{Definition}

\newcommand{\eps}{\varepsilon}

\newcommand{\poly}{\mathrm{poly}}

\newcommand{\I}[1]{\mathbb{I}\left[#1\right]}

\newcommand{\dist}{\text{dist}}

\newcommand{\OPT}{\mathrm{OPT}}

\newcommand{\R}{\mathbb{R}}
\newcommand{\norm}[1]{\left\lVert#1\right\rVert}
\newcommand{\inner}[1]{\left\langle #1 \right\rangle}

\newcommand{\proj}{\boldsymbol{\Pi}}
\newcommand{\projx}{\proj_{\xset}}
\newcommand{\projy}{\proj_{\yset}}
\newcommand{\defeq}{:=}
\newcommand{\diag}[1]{\textbf{\textup{diag}}\left(#1\right)}
\newcommand{\x}{^\mathsf{x}}
\newcommand{\y}{^\mathsf{y}}
\newcommand{\xset}{\mathcal{X}}
\newcommand{\yset}{\mathcal{Y}}
\newcommand{\Par}[1]{\left(#1\right)}

\newcommand{\fC}{\mathcal{C}}

\makeatletter
\newcommand{\algmargin}{\the\ALG@thistlm}
\makeatother
\algnewcommand{\parState}[1]{\State%
    \parbox[t]{\dimexpr\linewidth-\algmargin}{\strut\hangindent=\algorithmicindent \hangafter=1 #1\strut}}

\newcommand*{\alert}[1]{}
\newcommand{\goran}[1]{}
\newcommand{\christoph}[1]{}
\newcommand{\rasmus}[1]{}
\title{Acceleration for Distributed Transshipment and Parallel Maximum Flow}

\author{Christoph Grunau, Rasmus Kyng, Goran Zuzic}
\date{\today}

\begin{document}
\maketitle
\thispagestyle{empty}

\begin{abstract}
We combine several recent advancements to solve $(1+\varepsilon)$-transshipment and $(1+\varepsilon)$-maximum flow with a parallel algorithm with $\tilde{O}(1/\varepsilon)$ depth and $\tilde{O}(m/\varepsilon)$ work.
We achieve this by developing and deploying suitable parallel linear cost approximators in conjunction with an accelerated continuous optimization framework known as the box-simplex game by Jambulapati et al. (ICALP 2022).
A linear cost approximator is a linear operator that allows us to efficiently estimate the cost of the optimal solution to a given routing problem.
Obtaining accelerated $\varepsilon$ dependencies for both problems requires developing a stronger multicommodity cost approximator, one where cancellations between different commodities are disallowed. For maximum flow, we observe that a recent linear cost approximator due to Agarwal et al. (SODA 2024) can be augmented with additional parallel operations and achieve $\varepsilon^{-1}$ dependency via the box-simplex game.

\smallskip

For transshipment, we also construct a deterministic and distributed approximator. This yields a deterministic CONGEST algorithm that requires $\tilde{O}(\varepsilon^{-1}(D + \sqrt{n}))$ rounds on general networks of hop diameter $D$ and $\tilde{O}(\varepsilon^{-1}D)$ rounds on minor-free networks to compute a $(1+\varepsilon)$-approximation.
\end{abstract}

\thispagestyle{empty}

\newpage

\tableofcontents
\thispagestyle{empty}
\setcounter{page}{0}

\newpage

\section{Introduction}

The maximum flow and the shortest path problem are both fundamental tasks in algorithmic graph theory, serving as benchmarks for developing and evaluating algorithms. Solving these problems efficiently is critical for optimizing resources and operations in computer networks, logistics, and other systems. Their study leads to the development of key algorithm design techniques that are broadly applicable in computer science.

In this paper, we consider the $(1+\eps)$-approximate versions of these problems on undirected graphs in the scalable parallel and distributed settings.
By an approximate maximum flow, we mean a flow that routes given demands, while obtaining \emph{congestion} within $(1+\eps)$ of the minimum possible. Here the congestion of the flow is the maximum ratio between edge flow and edge capacity across all edges.
We solve a generalization of the $s$-$t$ shortest path problem known as \emph{transshipment},
where a flow is sought that routes given demands, while minimizing the sum of edge flows weighted by edge length, and we find a flow that minimizes quantity to within a factor $(1+\eps)$.
\rasmus{do we say something about rounding? can I call transshipment a generalization of shortest paths?}

Prior work has established scalable algorithms for these problems:
For maximum flow, very recent work~\cite{agarwal2024parallel} gives a parallel $\tilde{O}(\eps^{-3})$ depth and $\tilde{O}(\eps^{-3} m)$ work $(1+\eps)$-approximate algorithm, where $n := V(G), m := E(G)$ are the number of nodes and edges of the graph.
% In this paper we improve the dependency on the accuracy $\eps > 0$ to $\eps^{-1}$ for both problems in the parallel setting, as well as in the deterministic distributed setting for transshipment.
%For the purposes of this paper, approximating the solution means computing a feasible primal and dual solutions, corresponding to a flow $f \in \R^E$ and potential vector $\phi \in \R^V$ respectively, whose objectives are both approximately optimal.
\cite{2022sssp} gives a parallel $\tilde{O}(\eps^{-2})$ depth and $\tilde{O}(\eps^{-2} m)$ work algorithm\footnote{We use $\sim$, as in $\tilde{O}$, to hide $\poly(\log n)$ factors, where $n$ is the number of nodes in the graph.} for $(1+\eps)$-approximate transshipment.
We strengthen these results, giving the first parallel $\tilde{O}(\eps^{-1})$ depth and $\tilde{O}(\eps^{-1} m)$ work algorithms for undirected maximum flow and transshipment.
For maximum flow, we obtain the following result:

\begin{restatable}{theorem}{thmMaxflowPram}
\label{thm:maxflow-pram}
There exists a randomized PRAM algorithm that uses $\tilde{O}(\eps^{-1})$ depth and $\tilde{O}(\eps^{-1} m)$ work to solve $(1+\eps)$-approximate undirected maximum flow. The algorithm computes feasible primal $f \in \R^E$ and dual $\phi \in \R^V$ solutions whose objectives are $(1+\eps)$-approximately optimal.
The algorithm succeeds with high probability.
\end{restatable}

For transshipment, we give a stronger result: We develop a parallel algorithm which is deterministic and
extends to the distributed setting.
\begin{restatable}{theorem}{transshipmentdist}\label{thm:transshipment-distrib}
  There exists an algorithm that computes feasible primal and dual solutions whose objectives are $(1+\eps)$-approximately optimal in the following settings:
  \begin{itemize}
  \item In the deterministic PRAM setting with $\tilde{O}(\eps^{-1})$ depth and $\tilde{O}(\eps^{-1} m)$.
  \item In the deterministic CONGEST setting with $\tilde{O}(\eps^{-1} (D_G+\sqrt{n}) )$ rounds on the network $G$, where $D_G$ is the hop diameter of $G$. The round complexity can be improved to $\tilde{O}(\eps^{-1}D_G)$ rounds if $G$ comes from a fixed minor-free family.
  \item In the randomized HYBRID setting with $\tilde{O}(\eps^{-1})$ rounds.
  \end{itemize}
\end{restatable}

Key ingredients in recent algorithms parallelizing maximum flow and transshipment \cite{ASZ20,Li20,agarwal2024parallel} and earlier sequential algorithms \cite{She13,KLOS14,Pen16,She17,S17}
for these flow problems are
(1) a first-order continuous optimization approach and (2) a \emph{preconditioner} for the first-order method, based on so-called linear congestion or cost approximators (matrices that can be used to approximate the optimal solution).
This is also true of works solving transshipment in the distributed setting \cite{goranci2022universally, zuzic2021simple, 2022sssp}

%\rasmus{I'm worried that a reviewer would claim these results are implicit in JT23? Should we somehow grant that?}
Our result on maximum flow is conceptually simple. We show that the recent parallel cost approximator algorithm for maximum flow due to \cite{agarwal2024parallel} can be adapted for use in the accelerated first-order method of \cite{jambulapati2023revisiting}, i.e. the \emph{box-simplex game framework}.
Compared to unaccelerated first-order methods, the box-simplex game requires access to a slightly-strengthened set of operations over the cost approximator, and we show that this can be implemented in a very simple manner in parallel.

To give a randomized algorithm for parallel transshipment, we could (but choose not to) follow the same approach, and show that the cost approximator of \cite{Li20} allows us to efficiently and in parallel implement the operations required for the associated box-simplex game.
We instead prove a stronger result for transshipment, showing that we can implement the box-simplex game with a deterministic and distributed algorithm. To achieve this, we need to modify the linear cost approximator introduced in \cite{2022sssp}, as the existing construction does not allow us to efficiently parallelize the box-simplex game operations.

%Modern approaches to these problem frequently adopt a methodology centered on optimization. Indeed, theoretical state-of-the-art approaches to maximum flow and transshipment utilize interior point methods~\cite{DS08, M13, KLS20, vdBLLSSSW21, chen2022maximum}, gradient descent~\cite{She13,KLOS14,Pen16,She17,S17,agarwal2024parallel}, multiplicative weights~\cite{AroraHK12, christiano2011electrical}.

%Similarly, essentially all scalable parallel and distributed approaches to shortest path that are near-optimal first solve the seemingly harder transshipment problem, and then proceed to solve it using optimization-inspired approaches like multiplicative weights and gradient descent~\cite{ASZ20, goranci2022universally, zuzic2021simple, 2022sssp}. We optimize the $\eps$ dependency of these methods by applying acceleration via the framework of \emph{box-simplex games} developed and refined in a remarkable series of works \cite{She17, jambulapati2022regularized, assadi2022semi, jambulapati2023revisiting}.

%The box-simplex game framework essentially only requires the efficient construction and evaluation of the so-called \emph{linear cost approximators}, matrices that can be used to approximate the optimal solution (defined later).

Cost approximators have been constructed in prior works, yet their integration into the box-simplex framework for parallel and distributed settings remains unexplored. This is likely due complex machinery involved. Our work aims to benefit the parallel and distributed computing community by demonstrating the use of these tools. Additionally, the community has begun to build upon these unaccelerated results \cite{2022sssp,agarwal2024parallel}, and our contributions can serve as a black-box enhancement for many of these methods \cite{fox2023simple, schneider2023near}.
%On a technical level, we show that the linear cost approximators for maximum flow developed in \cite{agarwal2024parallel} can be directly applied, whereas those for transshipment~\cite{2022sssp} require minor modifications, which we detail in this paper.

\textbf{Transshipment and shortest path.} While the shortest path problem in the sequential setting has a simple and exact solution, the task becomes  harder in the parallel setting. In this paper we consider the $(1+\eps)$-approximate, undirected, and weighted version of the problem in the parallel setting and improve the state-of-the-art from $\tilde{O}(\eps^{-2} m)$ work and $\tilde{O}(\eps^{-2})$ depth to $\tilde{O}(\eps^{-1} m)$ work and $\tilde{O}(\eps^{-1})$ depth.

Essentially all modern approaches for the shortest path problem with near-optimal theoretical guarantees first solve the so-called \emph{transshipment problem}. Here, one is given an undirected weighted graph $G = (V, E, w)$ with \emph{edge weights} $w(e) > 0$ and a \emph{demand vector} $d \in \R^V$ where $d(v) > 0$ indicates a supply in $v$ and $d(v) < 0$ indicates a demand. It is guaranteed that $\sum_{v \in V} d(v) = 0$. The goal is to move the supply until its location matches the demand. See \Cref{fig:transshipment} for an example. Clearly, the transshipment generalizes the $s-t$ shortest path by setting $d = 1_s - 1_t \in \R^V$, where $1_x(v) = \I{x = v}$. For this reason, we focus on the transshipment problem directly instead of shortest path.
\alert{C: Maybe one should point out that approximate transshipment does not directly imply approximate s-t path or ssssp}

\begin{figure}[ht]
  \centering
  \begin{minipage}{.45\textwidth}
    \centering
    \begin{tikzpicture}[every node/.style={draw, circle}, >=latex]
      % Define nodes in a circular layout
      \node (A) at (45:2) {A};
      \node (B) at (90:2) {B};
      \node (C) at (135:2) {C};
      \node (D) at (180:2) {D};
      \node (E) at (225:2) {E};
      \node (F) at (270:2) {F};
      \node (G) at (315:2) {G};
      \node (H) at (0:2) {H};

      % Draw edges
      \draw (A) -- (B) -- (C) -- (D) -- (E) -- (F) -- (G) -- (H) -- (A);

      % Add labels for supply and demand
      \node[draw=none] at (45:2.8) {d(A) = +2};
      \node[draw=none] at (192:3.0) {d(D) = +1};
      \node[draw=none] at (90:2.7) {d(B) = -1};
      \node[draw=none] at (315:2.9) {d(G) = -1};
      \node[draw=none] at (225:2.9) {d(E) = -1};

      % Draw the flow
      \draw[->, very thick, blue] (A) -- (B);
      \draw[->, very thick, blue] (A) -- (H);
      \draw[->, very thick, blue] (H) -- (G);
      \draw[->, very thick, blue] (D) -- (E);
    \end{tikzpicture}
    \caption{Transshipment example: $G$ is a cycle with unit weights. There are 2 units of supply in A (indicated by a label $d(A) = +2$), and a unit of supply in D (indicated by $d(D) = +1$). There is a unit demand in B, G, and E (indicated by $d(\cdot) = -1$). The optimal flow is in blue and has a cost of 4.}
    \label{fig:transshipment}
  \end{minipage}%
  \hfill
  \begin{minipage}{.45\textwidth}
    \centering
    \begin{tikzpicture}[every node/.style={draw, circle}, >=latex]
      % Define nodes in a circular layout
      \node (A) at (45:2) {A};
      \node (B) at (90:2) {B};
      \node (C) at (135:2) {C};
      \node (D) at (180:2) {D};
      \node (E) at (225:2) {E};
      \node (F) at (270:2) {F};
      \node (G) at (315:2) {G};
      \node (H) at (0:2) {H};

      % Draw edges
      \draw (A) -- (B) -- (C) -- (D) -- (E) -- (F) -- (G) -- (H) -- (A);

      % Add labels for supply and demand
      \node[draw=none] at (45:2.8) {d(A) = -2};
      \node[draw=none] at (90:2.7) {d(B) = +2};
      \node[draw=none] at (168:3.0) {d(D) = +1};
      \node[draw=none] at (225:2.9) {d(E) = -1};

      % Draw the flow
      % Add thick blue edges with labels
      \draw[very thick, blue, ->] (D) -- (E) node[draw=none, midway, below, sloped, text=blue] {1};
      \draw[very thick, blue, ->] (B) -- (A) node[draw=none, midway, below, sloped, text=blue] {1.5};

      \draw[very thick, blue, <-] (A) to[bend right=45] (H) to[bend right=45] (G) to[bend right=45] (F) to[bend right=45] (E) to[bend right=45] (D) to[bend right=45] (C) to[bend right=45] (B) node[draw=none, midway, below, sloped, text=blue] {0.5};
    \end{tikzpicture}
    \caption{Maximum flow example: $G$ is a cycle with unit weights. The depicted flow sends $1$ unit from $D \to E$, $0.5$ units $B \to A$ along the small circle, and $1.5$ units $B \to A$ along the long circle. The solution has a cost of $1.5$ since, in our formulation, the most overcongested edge (either AB or DE) has congestion $1.5$. No better solution exists.}
    \label{fig:maxflow}
  \end{minipage}
\end{figure}

Formally, the transshipment problem, in the primal, can be formulated as finding a flow $f \in \R^E$ that satisfies the demand $d$ and has minimum cost. The dual can be formulated as finding the potentials $\phi \in \R^V$ that separates the positive and negative demand as much as possible while not overstretching edges. Formally, these can be written as follows.
\begin{align}
  \textbf{Primal:}~~ \min_{f \in \R^E}~ \norm{Wf}_1 : Bf = d, && \textbf{Dual:}~~ \max_{\phi \in \R^V}~ \inner{d, \phi} : \norm{W^{-1} B^\top \phi}_{\infty} \leq 1. \label{eq:TS-primal-dual}
\end{align}
Here, $W \in \R^{E \times E}$ is the diagonal matrix of edge weights (according to some fixed ordering $e_1, \ldots, e_m$) with $W_{e, e} = w(e)$. Furthermore, $B$ is the \emph{incidence matrix} of $G$ where the $i^{th}$ column corresponds to the $i^{th}$ edge $e_i$ with endpoints $a, b$, and the column is equal to $1_a - 1_b \in \mathbb{R}^V$. Note that the graph is undirected but with a fixed orientation of edges; a flow along the negative orientation is indicated with a negative value $f(e) < 0$.

\textbf{Maximum flow.} In our formulation of maximum flow, one is given a undirected and weighted graph $G = (V, E, w)$ where the capacity of an edge $e \in E$ is $1/w(e)$. In other words, the weight $w(e)$ corresponds to the inverse capacity. See \Cref{fig:maxflow} for an example. Furthermore, one is given a demand $d \in \mathbb{R}^V$ satisfying $\sum_{v \in V} d(v) = 0$ and the goal is to satisfy the demand $d$ with a flow $f$ that minimizes the most overcongested edge. Formally, this problem can be formulated as follows, reusing prior notation:
\begin{align}
  \textbf{Primal:}~~ \min_{f \in \R^E}~ \norm{Wf}_\infty : Bf = d, && \textbf{Dual:}~~ \max_{\phi \in \R^V}~ \inner{d, \phi} : \norm{W^{-1} B^\top \phi}_1 \leq 1. \label{eq:maxflow-primal-dual}
\end{align}

\textbf{Related work.}
Spielman and Teng \cite{ST04} developed a nearly-linear time algorithm
for solving the electrical flow problem, and ushered in an era of
using continuous optimization methods for graph algorithms and flow
problems in particular.
Daitch and Spielman used electrical flow solvers inside interior point
methods to develop algorithms for maximum and minimum-cost flow,
establishing the relevance of these methods for classic combinatorial
flow problems, but without attaining near-linear running times.
Christiano et al. and others \cite{christiano2011electrical,LRS13} used lower accuracy continuous
optimization approaches, in particular first-order approaches such as
gradient descent and multiplicative weight methods, to obtain further
speed-ups.
The first almost-linear time algorithms for approximate maximum flow
on undirected graphs were obtained by Sherman and by Kelner et
al. \cite{She13,KLOS14}, by leveraging $\ell_\infty$-norm gradient
descent, combined with $\ell_\infty$-oblivious routing or congestion
approximators \cite{Racke08,Madry10}.
Congestion approximators and oblivious routings in this setting act as
\emph{preconditioners} that radically reduce the iteration count of
the first-order methods being deployed.
By combining this approach recursively with a fast algorithm for
computing congestion approximators \cite{RST14}, Peng \cite{Pen16}  gave the first
$\tilde{O}(m\poly(\eps^{-1})))$ time algorithm for $(1+\epsilon)$
approximate maximum flow on undirected graphs.
First-order approaches were later extended to transshipment by Sherman
\cite{S17}, yielding the first $\tilde{O}(m\poly(\eps^{-1})))$ time algorithm for $(1+\epsilon)$
approximate transshipment on undirected graphs.
Sherman \cite{She17} developed a more sophisticated first-order
optimization method, based on a powerful concept of \emph{area-convexity},
and used this to give the first $\tilde{O}(m\eps^{-1}))$ time algorithm for
undirected maximum flow.
This was refined into a first-order method for both undirected maximum
flow and transshipment with running time $\tilde{O}(m\eps^{-1})$
by Jambulapati and co-authors \cite{jambulapati2022regularized,assadi2022semi,jambulapati2023revisiting},
using a framework they dubbed \emph{box-simplex games}.
Importantly, these accelerated first-order methods place slightly
stronger requirements on the preconditioner (typically a
congestion approximator) than earlier methods.

In the parallel settings, \cite{ASZ20} and \cite{Li20} gave the first $\tilde{O}(m\poly(1/\eps)))$
work algorithm for transshipment with $\tilde{O}(\poly(1/\eps))$
depth, building on the non-accelerated first-order approach introduced
by \cite{S17}. \cite{2022sssp} made the algorithm deterministic in both parallel and distributed settings. Very recently, \cite{agarwal2024parallel} gave the first $\tilde{O}(m\poly(1/\eps)))$ work algorithm for undirected maximum flow with $\tilde{O}(\poly(1/\eps)))$
depth, introducing a clever approach to parallelizing \cite{Pen16},
and deploying an unaccelerated first-order method as in
\cite{She13,KLOS14,Pen16}.

We omit a discussion of the extensive work on fast sequential algorithms for maximum flow and transshipment in directed graphs \cite{DS08, M13, KLS20, vdBLLSSSW21}, which recently culminated in an $m^{1+o(1)}$ time algorithm exact for minimum-cost flow \cite{chen2022maximum}, a problem that generalizes maximum flow and transshipment. The first sublinear and exact parallel and distributed shortest path algorithms on directed graphs was developed by \cite{rozhovn2022parallel,cao2023parallel}.

\section{Technical Overview}

The crux of maximum flow and transshipment algorithms that use the box-simplex game is to be able to construct and efficiently evaluate the so-called \emph{linear cost approximators} $R$. These are matrices with the property that the matrix-vector product $R d$ (where $d$ is a demand) has a norm $\norm{R d}$ that always approximates the cost required to optimally route a demand $d$. The approximators essentially serve as a gradient during the optimization process and are the main object of focus point in this paper.

\begin{definition}[Linear Cost Approximator]\label{def:cost-approx}
  Fix a weighted graph $G$.
  \begin{itemize}
  \item (Transshipment) A matrix $R \in \R^{ROWS(R) \times V}$ is an $\alpha$-apx linear cost approximator for transshipment if for all demand vectors $d \in \R^V$ with $\sum_{v \in V(G)} d(v) = 0$ we have that $\OPT(d) \le \norm{R d}_1 \le \alpha \cdot \OPT(d)$. Here $\OPT(d)$ is the optimum transshipment solution for a demand $d$.
  \item (Maximum flow) Similarly, an $\alpha$-apx linear cost approximator for maximum flow is a matrix $R \in \R^{ROWS(R) \times V}$ for which $\OPT(d) \le \norm{R d}_\infty \le \alpha \cdot \OPT(d)$ for all $d$ with $\sum_{v \in V(G)} d(v) = 0$. Again, $\OPT(d)$ is the optimum maximum flow solution for a demand $d$.
  \end{itemize}
\end{definition}

At a high-level, the box-simplex game is an iterative process that produces an approximate solution in $\tilde{O}(\eps^{-1} \alpha)$ iterations by using some $\alpha$-apx linear cost approximator $R$. In each iteration, it evaluates matrix-vector products $Ax, A^Tx, |A|x$, and $|A|^Tx$ for arbitrary vectors $x$, where $A := R B W^{-1}$ and $|A|$ is the entry-wise absolute value. Prior work, notably, also uses linear cost approximators, but utilizies a different optimization process that only requires evaluating $Rx$ and $R^T x$, resulting in a degraded $\eps$ dependency.

The main challenge in obtaining a $\eps^{-1}$ dependency is evaluating $|A| x$ and $|A|^T x$: simply being able to efficiently evaluate $R$ and $R^T$ is sufficient to efficiently evaluate $A x = R ( B ( W^{-1} x) )$ but not $|A| x = |R B W^{-1}| x$. To give context, the transshipment approximators in \cite{goranci2022universally,2022sssp} construct a \emph{dense} matrix $R \in \R^{E \times V}$ that maps a demand $d$ into a proper flow $f = Rd$ which routes $d$ and has approximately-optimal cost $\norm{R d}_1 \approx \OPT(d)$. This is a valid approximator, and the key insight to efficiently evaluating the matrix-vector products $R x, R^T x$ is by noting that, while dense, $R$ can be factored into two sparse matrices $R = R_1 R_2$. Here, both $R_1$ and $R_2$ have at most $\tilde{O}(m)$ entries (see remark in \cite[Sec 4.2 of the arXiv v2 version]{2022sssp}). This enables efficient evaluation of, say, $R x$ via $R_1 (R_2 x)$ in PRAM with $\tilde{O}(1)$ depth and $\tilde{O}(m)$ work. However, this does not give a direct way of evaluating $|R B W^{-1}| x$ as no similar factorization exists due to the application of entry-wise absolute values.

We provide an alternative, albeit very simple, observation: it is sufficient that the linear cost approximator $R$ is \emph{column sparse}. This enables one to explicitly construct $R B W^{-1}$ since, under the assumption, all of $R, B, W^{-1}$ are column sparse and this property is preserved under matrix multiplication. After constructing $R B W^{-1}$, one can directly take their absolute values. We note that this observation has also been implicitly used in \cite{jambulapati2023revisiting}, albeit they only state the result for the sequential setting.

\begin{definition}
  A matrix $R$ is column sparse (with respect to a parameter $n$) if each column has $\poly(\log n)$ non-zero elements.
\end{definition}
\begin{fact}
  Given two column sparse matrices $A, B$, their product $AB$ is also column sparse.
\end{fact}

Furthermore, we observe that the parallel maximum flow linear cost approximator obtained in very recent prior work \cite{agarwal2024parallel} is indeed column sparse and directly leads to a $\eps^{-1}$ dependency (\Cref{sec:maximum-flow-box-simplex}). On the other hand, the deterministic distributed transshipment linear cost approximator is not: we slightly change and adapt the construction to obtain a column sparse approximator. (\Cref{sec:transshipment-box-simplex} and \Cref{sec:det-distrib-transshipment-main}).

% \paragraph{Notation and assumptions.}
% For a graph $G$, we use $V(G)$ and $E(G)$ to denote its node set and edge set, respectively, and their sizes are denoted with $n := |V(G)|, m := |E(G)|$. It is often convenient to direct $E$ consistently. For simplicity and without loss of generality, we assume that $V = \{v_1, v_2, \ldots, v_{n}\}$ and define $\vec{E} = \{ (v_i,v_j) \mid (v_i, v_j) \in E, i < j\}$. We identify $E$ and $\vec{E}$ by the obvious bijection. We chose this orientation for simplicity and concreteness: arbitrarily changing the orientations does not influence the results (if done consistently). We also introduce the notation $\overleftrightarrow{E}$ which includes each edge of $\{u, v\} \in E$ in both directions: $(u, v) \in \overleftrightarrow{E}$ and $(v, u) \in \overleftrightarrow{E}$.  \goran{is $\overleftrightarrow{E}$ even used ?}

\goran{write combinatorial interpretation if there is time about flow cancellation and whatnot}

\textbf{Notation and assumptions.} Graphs are assumed to be connected, undirected, and weighted by default. Each graph $G$ comes with a function $w_G : E \rightarrow \{1, \ldots, n^{O(1)}\}$ that assigns a polynomially-bounded non-negative weight to each edge in $E$. This weight function induces a distance function $\dist_G(u, v)$ on $G$. For two sets of nodes $A,B \subseteq V(G)$ we define $\dist_G(U,W) = \min_{u \in U, w \in W} \dist_G(u, w)$ and sometimes write $\dist_G(U, w)$ for $\dist_G(U, \{w\})$.

\subsection{Preliminaries: Box-simplex game}\label{sec:box-simplex-prelim}

The box-simplex game is a family of optimization problems of the following form~\cite{jambulapati2023revisiting}:
\begin{align}
  \min_{x \in [-1, 1]^n} \max_{y \in \Delta_d} x^T A y - \inner{b, y} + \inner{c, x}. \label{eq:box-simplex}
\end{align}
Here, $\Delta_d := \{ x \in \mathbb{R}^d \mid x \ge 0, 1_d^T x = 1$ is the probability simplex, and $(A \in \mathbb{R}^{n \times d}, b \in \mathbb{R}^d, c \in \mathbb{R}^n)$ are given as input. We note the parameters $n$ and $d$ in this section are distinct from their meaning in the rest of the paper (outside of \Cref{sec:box-simplex-prelim}) to keep them in sync with prior work.

The value to which \Cref{eq:box-simplex} evalutes is referred to as the \emph{box-simplex game value} or the \emph{min-max value}. The goal is to find an $\eps$-approximate saddle point $(\hat{x}, \hat{y})$.
\begin{definition}
  Let $f(x, y) := x^T A y - \inner{b, y} + \inner{c, x}$. An $\eps$-approximate saddle point is a pair $(\hat{x}, \hat{y})$ when for all $x \in [-1, 1]^n$ and all $y \in \mathbb{R}^d, \norm{y}_1 \le 1$ we have that $f(x, \bar{y}) - f(\bar{x}, y) \le \eps$.
\end{definition}

\begin{algorithm}[h]
  \caption{\textsf{BoxSimplex}($\mathbf{A}, b, c, \eps$)~\cite{jambulapati2023revisiting}}
  \label{alg:boxsimplex}
  \begin{algorithmic}[1]
    \State {\bfseries Input:} $\mathbf{A}\in\R^{n \times d}$ with $L \defeq \norm{\mathbf{A}}_{1 \to 1}$, desired accuracy $\epsilon \in (0, L)$ \;
    \State {\bfseries Output:} an $\eps$-approxmate saddle point $(\hat{x}, \hat{y})$ to \Cref{eq:box-simplex}.
    \State $\Pi_{\mathbf{X}}(x)$ truncates each coordinate of $x$ to $[-1, +1]$. Formally: $x_i \mapsto \max(-1, \min(1, x_i))$.
    \State $\Pi_{\mathbf{Y}}(y) := \frac{y}{\norm{y}_1}$.
    \State $|A|$ takes the entry-wise absolute values, $(x_t)^2$ represents element-wise squaring, $x_t \odot y_t$ represent element-wise multiplication between two same-sized vectors, and $\frac{x}{y}$ is element-wise division for two equal-sized vectors.

    \State Initialize $x_0 \gets 0_n$, $y_0 \gets \frac 1 d 1_d$, $\bar{y}_0 \gets \frac 1 d 1_d$, $\hat{x} \gets 0_n$, $\hat{y} \gets 0_d$, $T \gets \lceil\frac{6(8\log d + 1)L}{\eps} \rceil$, $\alpha \gets 2$, $\beta \gets 2$. \;
    \State Rescale $\mathbf{A} \gets \frac 1 L \mathbf{A}$, $b \gets \frac 1 L b$, $c \gets \frac 1 L c$\;\label{line:rescale}
    \For{$t=0$ {\bfseries{\textup{to}}} $T-1$}
    \State $(g\x_t, g\y_t) \gets \frac 1 3 (\mathbf{A} y_t + c, b - \mathbf{A}^\top x_t)$ \Comment{Gradient oracle start.}\;
    \State $x^\star_t \gets \projx\Par{-\frac{g\x_t - 2\diag{x_t}|\mathbf{A}|y_t}{2|\mathbf{A}|y_t}}$ \;
    \State $y'_t \gets \projy\Par{y_t \circ \exp\Par{-\frac 1 \beta (g\y_t + |\mathbf{A}|^\top (x^\star_t)^2 - |\mathbf{A}|^\top x_t^2)}}$\;
    \State $x'_t \gets \projx\Par{-\frac{g\x_t - 2\diag{x_t}|\mathbf{A}|y_t}{2|\mathbf{A}|y'_t}}$
    \State $(\hat{x}, \hat{y}) \gets (\hat{x}, \hat{y}) + \frac 1 T (x'_t, y'_t)$    \Comment{Running average maintenance.}\;
    \State $(g\x_t, g\y_t) \gets \frac 1 6 (\mathbf{A} y'_t + c, b - \mathbf{A}^\top x'_t)$  \Comment{Extragradient oracle start.}\;
    \State $\bar{x}^\star_t \gets \projx\Par{-\frac{g\x_t - 2\diag{x_t}|\mathbf{A}|y_t}{2|\mathbf{A}|\bar{y}_t}}$ \;
    \State $y_{t + 1} \gets \projy\Par{\bar{y}_t \circ \exp\Par{-\frac 1 \beta (g\y_t + |\mathbf{A}|^\top (\bar{x}^\star_t)^2 + \alpha\log \bar{y}_t - |\mathbf{A}|^\top x_t^2 - \alpha \log y_t) }}$
    \State $x_{t + 1} \gets \projx\Par{-\frac{g\x_t - 2\diag{x_t}|\mathbf{A}|y_{t}}{2|\mathbf{A}|y_{t + 1}}}$ \;
    \State $\bar{y}_{t + 1} \gets \projy\Par{y_t \circ \exp\Par{-\frac 1 \beta (g\y_t + |\mathbf{A}|^\top (x_{t + 1})^2 + \alpha\log y_{t + 1} - |\mathbf{A}|^\top x_t^2 - \alpha \log y_t) }}$
    \EndFor
    \State {\bfseries Return:} $(\hat{x}, \hat{y})$
  \end{algorithmic}
\end{algorithm}
\goran{DOUBLE CHECK: $\alpha = \beta = 2$ in the above algo !??!?!?!?}

\begin{theorem}\label{thm:boxsimplex-iters}
  \Cref{alg:boxsimplex} computes an $\eps$-approximate saddle point $(\hat{x}, \hat{y})$ in $T := \frac{\norm{A}_{1 \to 1} \log d}{\eps}$ iterations. Here, $\norm{A}_{1 \to 1} := \max_{v \neq 0} \norm{A v}_1 / \norm{v}_1$ is the 1-operator norm.
\end{theorem}

Equivalent formulations of the box-simplex game, by minimizing over a single variable, are as follows.
\begin{proposition}
  Let $V$ be the game value of a specific box-simplex game specified by $(A, b, c)$. Then:
  \begin{align}
    V = - \min_{y \in \Delta_d} \big[ \norm{A y + c}_1 + \inner{b, y} \big] && \label{eq:maxsimplex} \\
    V = \min_{x \in [-1,1]^n} \big[ \max(A^T x - b) + \inner{c, x} \big] && \label{eq:minbox}
  \end{align}
  Moreover, subtituting $x$ and $y$ from an $\eps$-approximate saddle point $(x, y)$ into \Cref{eq:maxsimplex} and \Cref{eq:minbox} (resp.) yields values that are at least $V-\eps$ and at most $V+\eps$ (resp).
\end{proposition}
\begin{proof}
  Proven in Equation 1.2 of \cite{assadi2022semi} and the succeeding paragraph.
\end{proof}

Finally, we give an ``operational'' view of the box-simplex game optimizer described in \Cref{alg:boxsimplex}. Namely, the optimizer uses only 4 different classes of operations, as formalized in the following claim.
\begin{proposition}\label{prop:ops-from-boxsimplex}
  Each step of \Cref{alg:boxsimplex} can be implemented as $O(1)$ applications of the following operations:
  \begin{itemize}
  \item (O1) Initializing a vector in $\R^d$, $\R^n$, or $\R^1$ (i.e., scalar) to all zeros, or a value from the input.
  \item (O2) Applying an element-wise function $f : \R \to \R$ to each element of a previously computed vector.
  \item (O3) Computing dot products, or element-wise addition, multiplication, or division between two previously computed vectors of the same size.
  \item (O4) Computing matrix-vector products $A x$, $A^T x$, $|A| x$, $|A|^T x$ with a previously computed vector $x$ of the appropriate size. Here $|A|$ takes the element-wise absolute value.
  \end{itemize}
\end{proposition}
\begin{proof}
  By direct verification. For example, computing $g\x_t$ from line 10\goran{doublecheck if this is still true before submission} requires one to one to take the previously computed $y_t$, compute the matrix-vector product $A y_t$ (O4), add $c$ to it from the input (O1+O3), and multiply by $1/3$ (O2). Other operations are analogous.
\end{proof}

\subsection{Solving Maximum Flow with Box-Simplex}\label{sec:maximum-flow-box-simplex}

In this section we formally show how to construct an accelerated parallel algorithm that obtains feasible and approximate primal and solutions to the maximum flow via the box-simplex game.

First, we verify that indeed a column sparse $\tilde{O}(1)$-apx linear cost approximator was developed in \cite{agarwal2024parallel}.
\begin{lemma}\label{lemma:maxflow-routing}
  There exists a randomized PRAM algorithm with $\tilde{O}(1)$ depth and $\tilde{O}(m)$ work that, given a graph $G$, construct a $\poly(\log n)$-apx linear cost approximator that is also column sparse. The algorithm succeeds with high probability.
\end{lemma}
\begin{proof}
  Agarwal et al.~\cite{agarwal2024parallel} develop a linear cost approximator based on the R\"acke, Shah, T\"aubig~\cite{racke2014computing} cut-based hierarchical decomposition. The approximator has the following form. Let $T$ be a rooted tree of $O(\log n)$ height with $V(G) \subset V(T)$ and $|V(T)| = O(|V(G)|)$. The linear cost approximator matrix $R$ has the following form. Each row corresponds to an edge in the tree $T$. For a tree edge $e$ and vertex $v$, the entry $R_{e,v}$ is non-zero if and only if $e$ is on the path from $v$ to the root. Hence, column sparsity directly follows from tree height $O(\log n)$.
\end{proof}

Hence, we can use the box-simplex game solves maximum flow in $\tilde{O}(\eps^{-1})$ iterations, each iteration evaluating $R$ at most $O(1)$ times. We will construct the box-simplex instance by largely following the ideas from \cite{jambulapati2023revisiting,assadi2022semi}. First, Let $t > 0$ be a variable. Intuitively, $t$ corresponds to a guess of the optimal value $\OPT(d)$ for a particular instance of maximum flow. Then, set $A, b, c$ as follows.
\begin{align}
  A^T := \begin{bmatrix} RBW^{-1} \\ - RBW^{-1} \end{bmatrix}, \quad b := \begin{bmatrix} Rd \\ - Rd \end{bmatrix} / t, \quad c := 0 \label{eq:set-maxflow-bs-params}
\end{align}

First, we argue about the number of required iterations. Per \Cref{thm:boxsimplex-iters}, we require $\tilde{O}(\eps^{-1} \norm{A}_{1 \to 1})$ iterations. The following result stipulates that the box-simplex game requires $\tilde{O}(\eps^{-1})$ iterations.
\begin{lemma}
  Let $R$ be an $\poly(\log n)$-apx linear cost approximator for maximum flow and let $A^T$ be as in \Cref{eq:set-maxflow-bs-params}. Then $\norm{A}_{1 \to 1} \le \poly(\log n)$.
\end{lemma}
\begin{proof}
Note that, for a matrix $A$, $\norm{A}_{1 \to 1}$ is the maximum $\ell_1$ norm of a column. By standard linear algebra, this is equal to the $\infty$-operator norm of $A^T$, namely $\norm{A^T}_{\infty \to \infty} = \max_{\norm{f}_\infty \le 1} \norm{A^T f}_\infty$. However, by \Cref{def:cost-approx} we know that $\norm{A^T f}_\infty = \norm{R B W^{-1} f}_\infty \le \tilde{O}(1) \cdot \OPT(B W^{-1} f)$. However, $B W^{-1} f$ is exactly the demand vector that the flow $W^{-1} f$ satisfies, and that one has cost $\norm{W (W^{-1} f)}_\infty = \norm{f}_\infty$. Therefore, we conclude that $\norm{A^T f}_\infty \le \tilde{O}(1) \norm{f}_\infty$, hence $\norm{A^T}_{\infty \to \infty} \le \tilde{O}(1)$, hence $\norm{A}_{1 \to 1} \le \tilde{O}(1)$.
\end{proof}

We now ready to prove the main result of this subsection.
\thmMaxflowPram*
\begin{proof}
  For some parameter $t > 0$, we construct the box-simple game as above and obtain a $\eps$-approximate saddle point $(x, y)$. Then, we substitute $x := Wf/t$ for convenience. Now, we can rewrite \Cref{eq:minbox} like follows:
  \begin{align*}
    & \min_{x \in [-1,1]^n} \big[ \max(A^T x - b) + \inner{c, x} \big] \\
    = & \min_{Wf/t \in [-1,1]^n} \max(\begin{bmatrix} RBW^{-1} \\ - RBW^{-1} \end{bmatrix} Wf/t - \begin{bmatrix} Rd \\ - Rd \end{bmatrix} / t) \\
    = & \min_{\norm{Wf}_\infty \le t} \max(\begin{bmatrix} RB \\ - RB \end{bmatrix} f - \begin{bmatrix} Rd \\ - Rd \end{bmatrix}) / t \\
    = & \min_{\norm{Wf}_\infty \le t} \norm{RBf - Rd}_\infty / t  \numberthis \label{eq:maxflow-primal-game}
  \end{align*}
  Next, we show that the min-max game value is $\begin{cases}\text{exactly } 0 & \text{ if } t \ge \OPT(d) \\ \text{at least } (\OPT(d) - t)/t & \text{ if } t < \OPT(d) \end{cases}$. The first case is immediate: the value is non-negative due to \Cref{eq:maxflow-primal-game}, and the optimal primal solution $f^*$ achieves this value as $\norm{RBf^* - Rd}_\infty = \norm{Rd - Rd}_\infty = 0$. On the other hand, if $t < OPT(d)$, then the value of the game is $\norm{RBf - Rd}_\infty / t \ge OPT(d - Bf)/t \ge (OPT(d) - OPT(Bf)) / t$. However, $OPT(Bf) \le t$ since it is exactly the demand vector of the flow $f$ whose cost is at most $t$. Hence, $\norm{RBf - Rd}_\infty \ge (OPT(d) - t) / t$.

  We know that $\OPT(d)$ is polynomially bounded in $n$ since all weights $w(e)$ are polynomially bounded. This allows us to binary search in $\tilde{O}(1)$ steps the smallest value of $t$ where the value of the game is at most $2\eps$ (while solving the game up to accuracy $\eps$). Clearly, for this $t$, we have $t \le OPT(d)$. Furthermore, we also have $\OPT(d)/t - 1 - \eps \le 2\eps$, which can be rewritten as $t \ge (1 - O(\eps)) \OPT(d)$. Therefore, we know that $t$ is $(1+O(\eps))$-approximating $OPT(d)$.

  To obtain a feasible primal $(1+O(\eps))$-approximate solution, consider the final value $x$ and substitute $f = t \cdot W ^{-1} x$ like in \Cref{eq:maxflow-primal-game}. Unfortunately, such $f$ might not be feasible as in general $Bf \neq d$. However, let $d' := d - Bf$ be its ``error demand''. Since the game is $O(\eps)$-approximately solved, we have that $\OPT(d') \le \norm{R(d - Bf)}_\infty = \norm{RBf - Rd}_\infty < t \cdot O(\eps) < \OPT(d) O(\eps)$. Therefore, we can use a 2-approximate solution, say from \cite{agarwal2024parallel}, and obtain $f'$ such that $f + f'$ perfectly satisfies the demand (i.e., is feasible primal), and its cost is at most $\norm{Wf}_\infty + \norm{Wf'}_\infty \le t + 2 \OPT(d) O(\eps) \le (1 + O(\eps))\OPT(d)$.

  To obtain a feasible dual $(1+O(\eps))$-approximate solution, we solve the box-simplex game with the parameter $t' = (1-\eps) t > (1 - O(\eps))\OPT(d)$. Since $t$ was defined as the smallest threshold where the approximated game value is at most $2 \eps$, it must hold that the game value at $t' < t$ is at least $2\eps - \eps = \eps$. Hence, consider the final value $y = \begin{bmatrix} y_1 \\ y_2 \end{bmatrix}$ and substitute $\phi := - R^T (y_1 - y_2)$. Now, we can rewrite \Cref{eq:maxsimplex} as:
  \begin{align*}
    0 < \eps \le & - \big[ \norm{A y + c}_1 + \inner{b, y} \big] \\
    = & - \big[ \norm{\begin{bmatrix} W^{-1} B^T R^T & - W^{-1} B^T R^T \end{bmatrix} \begin{bmatrix} y_1 \\ y_2 \end{bmatrix}}_1 + \inner{\begin{bmatrix} Rd/t' \\ -Rd/t' \end{bmatrix}, \begin{bmatrix} y_1 \\ y_2 \end{bmatrix}} \big] \\
    = & - \big[ \norm{W^{-1} B^T \phi}_1 - \inner{ d, \phi } / t'\big] \\
    = & \inner{ d, \phi } / t' - \norm{W^{-1} B^T \phi}_1
  \end{align*}

  Hence, the inequality $\inner{ d, \phi } / t' - \norm{W^{-1} B^T \phi}_1 \ge 0$ implies $\inner{ d, \phi } \ge t' \cdot \norm{W^{-1} B^T \phi}_1$. Subtitute $\phi' := \phi / \norm{W^{-1} B^T \phi}_1$ to obtain our final dual solution. Clearly, $\norm{W^{-1} B^T \phi'}_1 = 1$ by construction. Furthermore, $\inner{ d, \phi' } \ge t' > (1-O(\eps)) \OPT(d)$, as required.
\end{proof}

\subsection{Solving Transshipment with Box-Simplex}\label{sec:transshipment-box-simplex}

In this section we explain how to obtain accelerated solutions for transshipment using the box-simplex game framework. This section closely follows the one for maximum flow \Cref{sec:maximum-flow-box-simplex} and largely borrows its ideas from \cite{assadi2022semi}.

First, for a fixed transshipment instance $(G, d)$, we introduce a parameter $t > 0$ and construct the box-simplex game in the following way.
\begin{align}
  A := \begin{bmatrix} R B W^{-1} & - R B W^{-1} \end{bmatrix}, \quad b := 0_E, \quad c := - \frac{1}{OPT} R d \label{eq:set-ts-bs-params}
\end{align}

To bound the number of iterations for transshipment via \Cref{thm:boxsimplex-iters} we need to show that $\norm{A}_{1 \to 1} \le \tilde{O}(1)$.
\begin{lemma}
  Let $R$ be an $\poly(\log n)$-approximate linear cost approximator of a graph $G$. Then $\norm{R B W^{-1}}_{1\to 1} \le \poly(\log n)$ and $\norm{A}_{1 \to 1} \le \poly(\log n)$.
\end{lemma}
\begin{proof}
  Note that, for a matrix $A$, $\norm{A}_{1 \to 1}$ is the maximum $\ell_1$ norm of a column. We now consider the $i$\textsuperscript{th} column of $R B W^{-1}$. The $i$\textsuperscript{th} column can be rewritten as $R \frac{1_{u} - 1_{v}}{w(e)}$ assuming the $i$\textsuperscript{th} column of $B$ corresponds to the edge $e = (u, v)$.

  We first note that, for each $e = (u, v)$, we have that $\OPT(1_{u} - 1_{v}) \le w(e)$ since this particular demand can be served by sending a unit flow over the edge $e$, achieving a cost of $w(e)$. Since flows and demands can be scaled, we have $\OPT(\frac{1_{u} - 1_{v}}{w(e)}) \le 1$. Hence, by \Cref{def:cost-approx} we have $\norm{ R \frac{1_{u} - 1_{v}}{w(e)} }_1 \le \tilde{O}(1)$. Since this holds for all $i$, we have that $\norm{R B W^{-1}}_{1 \to 1} \le \tilde{O}(1)$.
\end{proof}

With the above results in mind, we assert that one can obtain an accelerated transshipment solution if one can implement the four operations required in \Cref{prop:ops-from-boxsimplex}, along with a few other matrix-vector products required for post-processing.
\begin{proposition}
\label{prob:box_simplex_transshipment}
  Consider a transshipment instance $(G, d)$. Let $\eps > 0$ be the desired accuracy and $R$ be an $\alpha$-apx linear cost approximator. Then one can compute a feasible primal and dual solutions that are $(1+\eps)$-approximate with an algorithm that sequentially applies at most $O(\eps^{-1} \alpha \log m)$ operations. The operations are either (O1)--(O4) from \Cref{prop:ops-from-boxsimplex} with the substitution $A := R B W^{-1}$ for O4, or are matrix-vector products with $R, R^T, B, B^T, W, W^{-1}$.
\end{proposition}
\begin{proof}
  The proof exactly follows the one of \Cref{thm:maxflow-pram}, except that $\norm{\cdot}_1$ and $\norm{\cdot}_\infty$ are substituted one-for-one.
\end{proof}

\section{Deterministic Distributed Transshipment and Approximate Shortest Path}\label{sec:det-distrib-transshipment-main}
In this section, we prove our main theorem, which we restate here for convenience.
\transshipmentdist*.

As mentioned in the introduction, the result follows in a relatively straightforward manner by defining a \emph{column-sparse} linear cost approximator based on the linear cost approximator given in \cite{2022sssp}, which we then use inside the box-simplex game to efficiently solve approximate transshipment. We start by giving a high-level overview of the linear cost approximator of \cite{2022sssp}. For the following intuitive discussions, we assume that the graph $G$ is unweighted.

\textbf{Review of the linear cost approximator of \cite{2022sssp}.}
 The $\poly(\log n)$-approximate linear cost approximator $R_{obv}$ used in \cite{2022sssp} is a special kind of linear cost approximator that maps each demand into a valid flow; thus, it defines a so-called oblivious routing. The matrix $R_{obv} \in \mathbb{R}^{E \times V}$ can have up to $m \cdot n$ non-zero entries. Therefore, explicitly computing $R_{obv}$ would lead to algorithms with work $\Omega(mn)$. Instead, as mentioned in the introduction, \cite{2022sssp} efficiently compute matrix-vector products with $R_{obv}$ and $R_{obv}^T$ by factorizing the matrix $R_{obv} = R_1R_2$ into two sparse matrices with $R_1 \in \mathbb{R}^{E \times \mathcal{P}}$ and $R_2 \in \mathbb{R}^{\mathcal{P} \times V}$. Here, $\mathcal{P}$ is a set consisting of $\tilde{O}(n)$ directed paths in the graph $G$. Intuitively, $R_1$ maps a given directed path $P \in \mathcal{P}$ to a unit flow along that path. The set of paths satisfies the property that each edge is contained in at most $\tilde{O}(1)$ paths and therefore $R_1$ is row-sparse. The matrix $R_2$ maps each node $v$ to a set of $\tilde{O}(1)$ directed paths and is therefore column sparse.

\textbf{Deriving a column-sparse cost approximator from the one in \cite{2022sssp}.}
One simple way to derive a column-sparse cost approximator $R$ from $R_{obv}$ is as follows: Let $D_{\mathcal{P}} \in \mathbb{R} \in \mathbb{R}^{\mathcal{P} \times \mathcal{P}}$ be the diagonal matrix where the diagonal entry corresponding to path $P \in \mathcal{P}$ is set to the length of the path $P$. Now, let $ R =
 D_{\mathcal{P}}R_2 \in \mathbb{R}^{\mathcal{P} \times n}$. Intuitively, we obtain $R$ from $R_{obv}$ by replacing the matrix $R_1$ mapping a path to its edges with a matrix $D_{\mathcal{P}}$ mapping a path to its length. Column-sparsity of $R$ directly follows from column-sparsity of $D_{\mathcal{P}}$ and $R_2$, and a simple adaptation of the analysis of \cite{2022sssp} directly gives that $R$ is a $\poly(\log n)$-approximate cost approximator. In \Cref{sec:cost-approximator}, we define a slightly different cost approximator $R$, for simplicity reasons. In particular, to compute the matrix $R_2$, one has to compute the heavy-light decomposition of a given set of $O(\log^2 n)$ rooted forests (The set $\mathcal{P}$ contains the corresponding paths of the heavy-light decomposition). The computation of our cost approximator $R$ does not rely on heavy-light decompositions, which somewhat simplifies the distributed implementation of the box-simplex game.

\textbf{Roadmap.}
We express the distributed implementation of the box-simplex game in terms of the Minor-Aggregation model~\cite{goranci2022universally,ghaffari2022universally}, a convenient unifying interface with a simple ``programming API'' that can then be compiled into models like PRAM, CONGEST,  HYBRID, etc. In \Cref{sec:minor-aggregation}, we give a formal definition of the Minor-Aggregation model and state various simulation results for the aforementioned parallel and distributed models. In \Cref{sec:cost-approximator}, we first start by stating the relevant definitions in \cite{2022sssp} and then derive our cost approximator $R$. Finally, in \Cref{sec:distributed_implementation} we first discuss the distributed implementation of the various matrix-vector products required for the box-simplex game (and the post-processing steps), and then prove our main theorem.

\subsection{Preliminary: Minor-Aggregation Model}
\label{sec:minor-aggregation}
In this section, we describe the Minor-Aggregation model~\cite{goranci2022universally,ghaffari2022universally}, a powerful interface to design simple parallel and distributed graph algorithms.

\textbf{Aggregations.} An aggregation $\bigoplus$ is a binary operator that is commutative and associative (e.g., sum or max), which makes the value $\bigoplus_{i=1}^n m_i$ unique. More general definitions are possible but are not used in this paper~\cite{ghaffari2022universally}.

\begin{definition}[Distributed Minor-Aggregation Model]\label{def:aggregation-congest}
  We are given a connected undirected graph $G = (V, E)$. Both nodes and edges are individual computational units (i.e., have their own processor and private memory). Communication occurs in synchronous rounds. Initially, nodes only know their unique $\tilde{O}(1)$-bit ID and edges know the IDs of their endpoints. Each round consists of the following three steps (in that order).
  \begin{itemize}
  \item \textbf{Contraction step.} Each edge $e$ chooses a value $c_e = \{\bot, \top\}$. This defines a new \emph{minor network} $G' = (V', E')$ constructed as $G' = G / \{ e : c_e = \top \}$, i.e., by contracting all edges with $c_e = \top$ and self-loops removed. Vertices $V'$ of $G'$ are called supernodes, and we identify supernodes with the subset of nodes $V$ it consists of, i.e., if $s \in V'$ then $s \subseteq V$.

  \item \textbf{Consensus step.} Each node $v \in V$ chooses a $\tilde{O}(1)$-bit value $x_v$. For each supernode $s \in V'$, we define $y_s := \bigoplus_{v \in s} x_v$, where $\bigoplus$ is any pre-defined aggregation operator. All nodes $v \in s$ learn $y_s$.

  \item \textbf{Aggregation step.} Each edge $e \in E'$, connecting supernodes $a \in V'$ and $b \in V'$, learns $y_a$ and $y_b$ and chooses two $\tilde{O}(1)$-bit values $z_{e, a}, z_{e, b}$ (i.e., one value for each endpoint). For each supernode $s \in V'$, we define an aggregate of its incident edges in $E'$, namely $\bigotimes_{e \in \text{incidentEdges(s)}} z_{e, s}$ where $\bigotimes$ is some pre-defined aggregation operator. All nodes $v \in s$ learn the aggregate value (they learn the same aggregate value, a non-trivial assertion if there are many valid aggregates).
  \end{itemize}
\end{definition}

Note that the input weights (when solving maximum flow or transshipment) do not affect the communication in any way.

\textbf{Parallel and distributed implementations.} Any algorithm implemented in the Minor-Aggregation model can be efficiently simulated in the parallel PRAM model and in various distributed models.

\begin{fact}
  \label{fact:pram_simulation}
  A $T$-round Minor-Aggregation algorithm can be simulated in PRAM with $\tilde{O}(T)$ depth and $\tilde{O}(T \cdot m)$ work.
\end{fact}
\begin{proof}
  Proven in Theorem B.9 from \cite[arXiv v2 version]{2022sssp}.
\end{proof}

\begin{fact}
\label{fact:congest_simulation}
  A deterministic $T$-round Minor-Aggregation algorithm can be simulated by a deterministic CONGEST algorithm running in $\tilde{O}(T(D_G + \sqrt{n}))$ where $D_G$ is the hop diameter of the network $G$. The round complexity of the simulation can be improved to $\tilde{O}(T D_G)$ if $G$ comes from a fixed minor-free family.
\end{fact}
\begin{proof}
  Theorem B.9 from \cite[arXiv v2 version]{2022sssp} proves that a $T$-round Minor-Aggregation algorithm can be simulated with $\tilde{O}(T)$-round in $\text{CONGEST}^{PA}$ (using the notation from the paper). In turn, Theorem B.3 shows that in regular CONGEST, such an algorithm can be simulated:
  \begin{itemize}
  \item In deterministic $\tilde{O}(D_G + \sqrt{n})$ rounds on a arbitrary network $G$ that has hop diameter $D_G$.
  \item In deterministic $\tilde{O}(D_G)$ rounds on a minor-free network network $G$ that has hop diameter $D_G$.
  \end{itemize}
\end{proof}

Very recent work~\cite{schneider2023near} has also extended the Minor-Aggregation interface to new models.
\begin{fact}
  \label{fact:hybrid_simulation}
  A $T$-round Minor-Aggregation algorithm can be simulated with $\tilde{O}(T)$ rounds in the HYBRID model, the simulation succeeds with high probability.
\end{fact}
\begin{proof}
  Proven in Lemma 7 of \cite[arXiv v2 version]{schneider2023near}.
\end{proof}

% \alert{Should we remove this?}
% \begin{fact}
%   A $T$-round Minor-Aggregation algorithm can be simulated with $\tilde{O}(T)$ rounds in the Supported-CONGEST model.
% \end{fact}

% \goran{is this needed + connective text if yes}
% \begin{proposition}
%   For a transshipemnt problem, let $A := R B W^{-1}$. If one can evaluate matrix-vector products $Ax, A^Tx,|A|x, |A^T|x, Rx, R^Tx$ for all $x$, then one can obtain a feasible and $(1+\eps)$-approximate primal and dual solution.
%   \goran{in how many rounds?}
% \end{proposition}

\subsection{Column-Sparse Cost Approximator}
\label{sec:cost-approximator}
\subsubsection{Relevant definitions from \cite{2022sssp}}

We start by giving the definitions in \cite{2022sssp} that are needed to describe the linear cost approximator. We omit an intuitive discussion of the definitions below and refer the interested reader to the corresponding parts in \cite{2022sssp}.

\begin{definition}[Cluster, Clustering, Sparse Neighborhood Cover (Section 2 in the arXiv v2 version of \cite{2022sssp})]
A \emph{cluster} $C$ is a set of nodes $C \subseteq V(G)$.
A \emph{clustering} $\fC$ is a collection of disjoint clusters.
Finally, a \emph{sparse neighborhood cover} of a graph $G$ with \emph{covering radius} $R$ is a collection of $\gamma := O(\log n)$ \footnote{Sometimes we use $\log(\cdot)$ outside the $O$-notation, in which case it stands for $\log_2(\cdot)$. } clusterings $\fC_1, \dots, \fC_\gamma$ such that for each node $v \in V(G)$ there exists some $i \in \{1,\ldots,\gamma\}$ and some $C \in \fC_i$ with $B(v, R) := \{u \in V(G) \colon d_G(u,v) \leq R\} \subseteq C$.
\end{definition}

\begin{definition}[Distance Scales (Definition 3.2 in the arXiv v2 version of \cite{2022sssp})]
We set $\tau = \log^7(n), \beta = 8\tau$ and $D_i = \beta^i$ and $I_{scales} := \{i \in \mathbb{N}_0,D_i \leq n^2 \max_{e \in E}w(e)\}\}$ and $i_{max} := \max I_{scale}$
\end{definition}

\begin{definition}[Distance Structure for Scale $D_i$ (Definitions 3.5 and 3.6 in the arXiv v2 version of \cite{2022sssp})]
\label{def:distance_structure}
A distance structure for scale $D_i$ consists of a sparse neighborhood cover covering radius $\frac{D_i}{\tau}$ and where each cluster $C$ in the sparse neighborhood cover has a diameter of at most $D_i$, meaning that $\max_{u,v \in C} d_{G[C]}(u,v) \leq D_i$. Moreover, each cluster $C$ comes with a mapping $\phi_{V \setminus C} \colon C \mapsto \mathbb{R}_{\geq 0}$ satisfying
\begin{enumerate}
    \item $\forall v \in C$, $\phi_{V \setminus C}(v) \leq \min(d_G(V \setminus C, v), D_i)$
    \item $\forall u,v \in C, |\phi_{V \setminus C}(u) - \phi_{V \setminus C}(v)|\leq d_G(u,v)$
    \item $d_G(V \setminus C,v) \geq \frac{D_i}{\gamma}$ implies $\phi_{V \setminus C}(v) \geq \frac{D_i}{2\tau}$.
\end{enumerate}
\end{definition}

\cite{2022sssp} gives a distributed algorithm for computing a distance structure for every scale $D_i$. The algorithm runs in $\tilde{O}(1)$ Minor-Aggregation rounds and additionally assumes access to an oracle for computing Eulerian Orientations (see Definition 3.8 in \cite{2022sssp} for the formal oracle definition). The formal statement follows by combining Lemmas 3.14 and 3.15 in \cite{2022sssp}. Together with the Minor-Aggregation simulation results given in \Cref{sec:minor-aggregation} and efficient parallel and distributed algorithms for computing Eulerian Orientations (Theorem 6.1 in \cite{2022sssp}, \cite{atallah_vishkin1984euler_pram} and Lemma 10 in \cite{schneider2023near}), one obtains the following result.

\begin{theorem}[Distributed Computation of Distance Structures (\cite{2022sssp}, \cite{atallah_vishkin1984euler_pram}, \cite{schneider2023near})]\label{thm:distance_structure_distributed}
  There exists an algorithm that computes a distance structure for scale $D_i$ for every $i \in I_{scale}$ in the following settings:
  \begin{itemize}
  \item There exists a deterministic PRAM algorithm using $\tilde{O}(\eps^{-1})$ depth and $\tilde{O}(\eps^{-1} m)$ achieving the guarantees.
  \item There exists an $\tilde{O}(\eps^{-1} (D_G+\sqrt{n}) )$-round deterministic algorithm in CONGEST on the network $G$ achieving the guarantees, where $D_G$ is the hop diameter of the network $G$. The round complexity can be improved to $\tilde{O}(\eps^{-1}D_G)$ rounds if $G$ comes from a fixed minor-free family
  \item There exists a randomized $\tilde{O}(\eps^{-1})$-round HYBRID algorithm achieving the guarantees.
  \end{itemize}
\end{theorem}

\subsubsection{Cost Approximator from Distance Structures}
We assume that we are given a distance structure for scale $D_i$ for every $i \in I_{scale}$. We also assume that the sparse neighborhood cover with diameter $D_0 = 1$ and covering radius $D_0/\tau < 1$ only consists of singleton clusterings. We next define a $\poly(\log n)$-approximate linear cost approximator based on the distance structure for every scale $D_i$.
To that end, we first start by introducing some additional notation before defining the cost approximator matrix $R$.
We denote by $\mathcal{S}_i$ the sparse neighborhood cover of the distance structure for scale $D_i$. We assume that each sparse neighborhood cover consists of exactly $NUM = O(\log n)$ clusterings. For $j \in [NUM]$, we denote by $\fC_{i,j}$ the $j$-th clustering of $\mathcal{S}_i$ and we also assume (without loss of generality) that each node in $V$ is clustered in $\fC_{i,j}$. For a node $v$, we denote by $C_{i,j}(v)$ the cluster
 in $\fC_{i,j}$ that contains $v$. Also, recall that each cluster $C \in \fC_{i,j}$ comes with a potential $\phi_{V \setminus C} \colon C \mapsto \mathbb{R}_{\geq 0}$. We define $\phi_{i,j} \colon V \mapsto \mathbb{R}_{\geq 0}$ with $\phi_{i,j}(v) = \phi_{V \setminus C_{i,j}(v)}(v)$ for every $v \in V$. For every $i \in I_{scale}$, $j \in [NUM]$ and $v \in V$, we define

\[p_{i,j}(v) := \max(0,\phi_{i,j}(v)/D_i - 0.25/\tau) \text{ and } w_i(v) = \sum_{j=1}^{NUM} p_{i,j}(v).\]

We are now ready to define the linear cost approximator. First, we index the rows in $R$ using the following set

\[ROWS(R) := \{(i,j,j',C) \colon i \in I_{scale} \setminus \{i_{max}\}, j,j' \in [NUM], C \in \fC_{i,j}\}.\]

That is, $R$ maps any node demand $d \in \mathbb{R}^V$ to a vector $Rd \in \mathbb{R}^{ROWS(R)}$.

Note that $|I_{scale}|,NUM \in O(\log n)$ and $|\fC_{i,j}| \leq n$. Thus, $|ROWS(R)| = O(n \log^3(n))$.

For any tuple $(i,j,j',C) \in ROWS(R)$ and $v \in V$, we define

\[R_{(i,j,j',C),v} :=\begin{cases}
  D_{i+1} \cdot\frac{p_{i,j}(v) \cdot p_{i+1,j'}(v)}{w_i(v)\cdot w_{i+1}(v)} &  \text{ if $v$ is in $C$ (so $C_{i,j}(v) = C$) and $C \subseteq C_{i+1,j'}(v)$} \\
  0 & \text{ otherwise}
\end{cases}.\]
The fact that $R$ is a $\poly(\log n)$-approximate linear cost approximator for transshipment follows by a straightforward adaptation of the proof in Section $4$ of \cite{2022sssp} that shows that their oblivious routing defines a $\poly(\log n)$-approximate linear cost approximator.

\begin{lemma}[$R$ is a $\poly(\log n)$-approximate cost approximator (See Section $4$ in the arXiv v2 version of \cite{2022sssp})]
\label{lem:R_is_cost_approximator}
The matrix $R \in \mathbb{R}^{ROWS(R) \times V}$ defines a $\poly(\log n)$-approximate linear cost approximator for transshipment.
\end{lemma}
For the sake of brevity, we omit a formal proof and instead explain how our cost approximator relates to the oblivious routing defined in \cite{2022sssp}.
The oblivious routing of \cite{2022sssp} defines a flow $f_v$ for every node $v \in V$ (the column corresponding to $v$). The flow $f_v$ can be decomposed into $f_v = \sum_{i,j,j'} f_{v,(i,j,j')}$ where $f_{v,(i,j,j')}$ is a flow that sends $\frac{p_{i,j}(v) \cdot p_{i+1,j'}(v)}{w_i(v)\cdot w_{i+1}(v)}$ units of flow from the cluster center of $C_{i,j}(v)$ to the cluster center of $C_{i+1,j'}(v)$ along a path of length at most $D_{i+1}$. In particular, \cite{2022sssp} assumes that the cluster $C_{i+1,j'}(v)$ comes with a spanning tree of diameter at most $D_{i+1}$ and the flow path is the unique path connecting the cluster center of $C_{i,j}(v)$ with the cluster center of $C_{i+1,j'}(v)$. Note that this path only exists if the cluster center of $C_{i,j}(v)$ is actually contained in $C_{i+1,j'}(v)$. Indeed, $f_{v,(i,j,j')}$ only sends a non-zero amount of flow if $p_{i+1,j'}(v) > 0$, in which case the even stronger property $C_{i,j}(v) \subseteq C_{i+1,j'}(v)$ holds. Our cost approximator $R$ overapproximates the cost incurred by $f_{v,(i,j,j')}$ sending $\frac{p_{i,j}(v) \cdot p_{i+1,j'}(v)}{w_i(v)\cdot w_{i+1}(v)}$ units of flow along a path of length at most $D_{i+1}$ by setting $R_{(i,j,j',C_{i,j}(v)),v} = D_{i+1} \cdot \frac{p_{i,j}(v) \cdot p_{i+1,j'}(v)}{w_i(v)\cdot w_{i+1}(v)}$. It also overapproximates the cost of the flow by omitting certain cancellations (for example if $f_{v,(i_1,j_1,j_1')}$ and $f_{v,(i_2,j_2,j_2')}$ send a flow in opposite directions along a given edge). However, the analysis in Section 4 of \cite{2022sssp} does not rely on those cancellations and thus \cref{lem:R_is_cost_approximator} indeed follows by a straightforward adaptation of the analysis given in \cite{2022sssp}.

\subsection{Distributed Implementation}
\label{sec:distributed_implementation}

It remains to show that we can efficiently implement the box-simplex game in the Minor-Aggregation model using the column-sparse cost approximator $R$ defined in the previous section.
\subsubsection{Computing matrices $R$ and $A := RBW^{-1}$}
\begin{claim}[Computing $R$ and $A := RBW^{-1}$, given distance structures]
\label{claim:computingRandA}
Assume one is given access to the distance structures defining the cost approximator $R$. Then, in $\tilde{O}(1)$ Minor-Aggregation rounds each node $v$ can compute all the $\tilde{O}(1)$ non-zero entries in the column of $R$ corresponding to $v$ and each edge $e$ can compute all the $\tilde{O}(1)$ non-zero entries in the column of $A := RBW^{-1}$ corresponding to $e$.
\end{claim}
\begin{proof}
We first discuss how each node $v$ can efficiently learn the non-zero entries in the column of $R$ corresponding to $v$. First, each node $v$ can locally compute $p_{i,j}(v)$ and $w_{i}(v)$ for every $i \in I_{scale}$ and $j \in [NUM]$. Thus, the only obstacle for computing its non-zero entries is for $v$ to learn whether the cluster $C_{i,j}(v)$ is fully contained in the cluster $C_{i+1,j'}(v)$ for all $i \in I_{scale} \setminus \{i_{max}\}$ and $j,j' \in [NUM]$. For a given $i,j,j'$, each node $v$  can learn whether the cluster $C_{i,j}(v)$ is fully contained in the cluster $C_{i+1,j'}(v)$ in $O(1)$ Minor-Aggregation rounds, using the fact that the cluster $C_{i,j}(v)$ is connected. In particular, for each node $u$, let $x_u$ denote the unique identifier of the cluster $C_{i+1,j'}(u)$. Now, in $O(1)$ Minor-Aggregation rounds, we can learn whether the $x_u$'s are equal for all $u \in C_{i,j}(v)$, and therefore whether $C_{i,j}(v) \subseteq C_{i+1,j'}(v)$. As there are $O(\log^3 n)$ choices for $(i,j,j')$, each node $v$ can indeed compute all the $O(\log^3 n)$ non-zero entries in the column of $R$ corresponding to $v$. Thus, it remains to discuss how each edge $e \in E$ can compute all the $\tilde{O}(1)$ non-zero entries in the column of $A := RBW^{-1}$ corresponding to $e$. This is immediate, as it suffices for each edge $e = \{u,v\}$ to learn the $\tilde{O}(1)$ non-zero entries in $R$ in the two columns corresponding to $u$ and $v$, which can trivially be done in $\tilde{O}(1)$ Minor-Aggregation rounds.
\end{proof}

\subsubsection{Evaluating Matrix-Vector Products}
We next discuss how to efficiently evaluate all the matrix-vector products required for the box-simplex game. A vector $x \in \mathbb{R}^{ROWS(R)}$ is stored distributedly if each node $v \in V$ knows the entry for each $(i,j,j',C) \in ROWS(R)$ with $v \in C$.
\begin{claim}[Matrix-Vector Products]
    \label{claim:matrix_vector_products}
    Assume one is given access to the distance structures defining the cost approximator $R$. Then, in $\tilde{O}(1)$ Minor-Aggregation rounds one can compute matrix-vector products with $R, R^T, A, A^T, |A|$ and $|A^T|$.
\end{claim}
\begin{proof}
Using \Cref{claim:computingRandA}, we can assume that each node $v$ can compute all the $\tilde{O}(1)$ non-zero entries in the column of $R$ corresponding to $v$ and each edge $e$ can compute all the $\tilde{O}(1)$ non-zero entries in the column of $A := RBW^{-1}$ corresponding to $e$. With this information, given a vector $y \in \mathbb{R}^{ROWS(R)}$, computing $R^Ty \in \mathbb{R}^V$, $A^Ty \in \mathbb{R}^E$ and $|A^T|y \in \mathbb{R}^E$ simply boils down to locally computing dot-products, which can be done without any further communication round. We next discuss how to compute $Rx \in \mathbb{R}^{ROWS(R)}$ for a given vector $x \in \mathbb{R}^{V}$. Fix any $i \in I_{scale} \setminus \{i_{max}\}, j,j' \in [NUM]$. Using the fact that each cluster in $\fC_{i,j}$ is connected and $R_{(i,j,j',C),v}$ can only be non-zero if $v \in C$, we can compute in $O(1)$ Minor-Aggregation rounds $(Rx)_{(i,j,j',C)}$ for every $C \in \fC_{i,j}$, as $\sum_{u \in C} R_{(i,j,j',C),u}x_u$ can be computed using a single aggregation. Thus, we can indeed compute $Rx$ in $\tilde{O}(1)$ Minor-Aggregation rounds, using the fact that there are $O(\log^3 n)$ choice for $(i,j,j')$. In a similar vein, given $x \in \mathbb{R}^E$, we can compute $Ax$ and $|A|x$ in $\tilde{O}(1)$ Minor-Aggregation rounds.
\end{proof}

\subsubsection{$(1+\eps)$-Approximate
Primal and Dual Transshipment Solutions given Distance Structures, and Proof of our Main Theorem}

As we are able to efficiently evaluate the various matrix-vector products, we can use the box-simplex framework to efficiently solve $(1+\eps)$-transshipment with an improved $\eps$-dependency of $\eps^{-1}$.

\begin{lemma}[Distributed Box-Simplex for Transshipment]
\label{lem:distributed_box_simplex}
 Assume one is given access to the distance structures defining the cost approximator $R$ and a feasible transshipment demand $d$. Then, in $\tilde{O}(1/\eps)$ Minor-Aggregation rounds one can compute feasible primal and dual transshipment solutions that are $(1+\eps)$-approximate.
\end{lemma}
\begin{proof}
This directly follows by combining \Cref{claim:matrix_vector_products}, \Cref{lem:R_is_cost_approximator} and \Cref{prob:box_simplex_transshipment}.
\end{proof}

We are now ready to prove our main theorem.

\begin{proof}[Proof of \Cref{thm:transshipment-distrib}]
We first compute a distance structure for scale $D_i$ for every $i \in I_{scale}$ using the PRAM/CONGEST/HYBRID algorithms of \Cref{thm:distance_structure_distributed}. Then, given the distance structures, \Cref{lem:distributed_box_simplex} states that we can compute $(1+\eps)$-approximate primal and dual transshipment solutions via the box-simplex game in $\tilde{O}(1/\eps)$ Minor-Aggregation rounds. The theorem then follows by the Minor-Agggregation simulation results stated in \Cref{sec:minor-aggregation}.
\end{proof}

\bibliographystyle{alpha}
\bibliography{refs}

\end{document}